\documentclass[aoas]{imsart}

\usepackage{graphicx}
\usepackage{amsmath}
\usepackage{mathtools}
\usepackage{amsthm}                
\usepackage{amssymb}
\usepackage{amsfonts}
\usepackage[authoryear]{natbib}
\RequirePackage[colorlinks, citecolor=blue, urlcolor=blue]{hyperref}
\usepackage{wrapfig}
\usepackage{multirow}
\usepackage{tabularx}
\usepackage{caption}
\usepackage{mathrsfs}
\usepackage{algorithm}
\usepackage{booktabs,fixltx2e}
\usepackage{appendix}
\usepackage{float}

\startlocaldefs
\theoremstyle{plain}

\newtheorem{theorem}{Theorem}[section]
\newtheorem{lemma}[theorem]{Lemma}
\newtheorem{condition}[theorem]{Condition}

\theoremstyle{remark}
\newtheorem{definition}[theorem]{Definition}
\newtheorem{remark}[theorem]{Remark}
\newtheorem{example}{Example}

\newtheorem*{example*}{Motivating Example}
\usepackage[T1]{fontenc}
\usepackage[utf8]{inputenc}

\newcommand{\R}{\mathbb R}

\newcommand{\w}{\mathbf{w}}

\newcommand{\va}{\boldsymbol{a}}
\newcommand{\one}{\boldsymbol{1}}

\newcommand{\vb}{\mathbf{b}}

\newcommand{\E}{\text{E}}

\newcommand{\vectorize}{ \text{vec} }
\newcommand{\var}{\text{var}}

\newcommand{\f}{f^0}

\newcommand{\sH}{\mathcal{H}}
\newcommand{\tk}{t^{(k)}}
\newcommand{\sS}{\mathcal{S}}

\newcommand{\half}{\frac{1}{2}}
\newcommand{\Sn}{\Sigma_n}

\newcommand{\Sninvhalf}{\Sigma_n^{-\frac{1}{2}}}

\newcommand{\St}{\Sigma}
\newcommand{\Sinv}{\Sigma^{-1}}
\newcommand{\Shalf}{\Sigma^{\frac{1}{2}}}
\newcommand{\Sinvhalf}{\Sigma^{-\frac{1}{2}}}

\newcommand{\ldn}{\Lambda_n}
\newcommand{\ld}{\Lambda}

\newcommand{\Nk}{\mathbf{N}_k}
\newcommand{\N}{\mathbf{N}}
\newcommand{\matN}{\mathbb N}
\newcommand{\matZ}{\mathbb Z}

\newcommand{\vv}{\mathbf{v}}

\def\red#1{\textcolor{black}{#1}}

\endlocaldefs
\begin{document}

\begin{frontmatter}
\title{B-scaling: A Novel Nonparametric Data Fusion Method}
\runtitle{B-scaling: A Novel Nonparametric Data Fusion Method}

\begin{aug}
\author[A]{\fnms{Yiwen}    \snm{Liu} \ead[label=e1,mark]{yiwenliu@arizona.edu}},
\author[A]{\fnms{Xiaoxiao} \snm{Sun} \ead[label=e2,mark]{xiaosun@arizona.edu}},
\author[B]{\fnms{Wenxuan}  \snm{Zhong}\ead[label=e3,mark]{wenxuan@uga.edu}}
\and
\author[C]{\fnms{Bing} \snm{Li}\ead[label=e4]{bing@stat.psu.edu}}
\address[A]{Department of Epidemiology and Biostatistics,
University of Arizona,
\printead{e1,e2}}

\address[B]{Department of Statistics,
University of Georgia,
\printead{e3}}

\address[C]{Department of Statistics,
Pennsylvania State University,
\printead{e4}}
\end{aug}






\begin{abstract}
Very often for the same scientific question, there may exist different techniques or experiments that measure the same numerical quantity. 
Historically, various methods have been developed to exploit the information within each type of data independently. However, statistical data fusion methods that could effectively integrate multi-source data under a unified framework are lacking. 
In this paper, we propose a novel data fusion method, called B-scaling, for integrating multi-source data. Consider $K$ measurements that are generated from different sources but measure the same latent variable through some linear or nonlinear ways. We seek to find a representation of the latent variable, named B-mean, which captures the common information contained in the $K$ measurements while takes into account the nonlinear mappings between them and the latent variable.
We also establish the asymptotic property of the B-mean and apply the proposed method to integrate multiple histone modifications and DNA methylation levels for characterizing epigenomic landscape. Both numerical and empirical studies show that B-scaling is a powerful data fusion method with broad applications.
\end{abstract}


\end{frontmatter}

\section{Introduction}

Both the amount and variety of data have been increasing dramatically in recent years from all fields of science, such as genomics, chemistry, geophysics, and engineering. Even for the same scientific question, there may exist different techniques or experiments that measure the same numerical quantity. Historically, various methods have been developed to exploit the information within each type of data independently. However, the development of methods that could effectively integrate multi-source data under a unified framework is lagging behind \citep{ritchie2015methods}. In contrast, such multi-source data are abundant in practice. One motivating example is the Roadmap Epigenomics Project \citep{bernstein2010nih}.

\begin{example*}[Epigenetic Data from the Roadmap Epigenomics Project]

As an emerging science, epigenomics is the study of the complete set of epigenetic modifications that regulate gene expressions \citep{egger2004epigenetics,bird2007perceptions,esteller2008epigenetics}.
With advanced sequencing technology, the Roadmap Epigenomics Project has provided a large amount of epigenomic data of different types and structures to characterize the epigenomic landscapes of primary human tissues and cells under different conditions or diseases \citep{romanoski_epigenomics:_2015,kundaje_integrative_2015,amin2015epigenomic}. Chromatin immunoprecipitation (ChIP) and bisulfite treatment, for example, are two different techniques. They are adopted to generate datasets of various types of histone modification levels and DNA methylation levels (Figure \ref{fig:epiindex}), all of which are indicators of epigenetic activeness. Despite the availability of large amounts of data, there is still a lack of systematic understanding of how epigenomic variations across genome relate to gene expression changes \citep{ritchie2015methods, gomez2014data, hamid2009data}. To bridge this gap is the direct motivation of our data fusion method. We illustrated in Figure \ref{fig:epiindex} the application of data fusion for epigenetic data. The goal is to fuse multiple histone modifications and DNA methylation levels and develop an epigenetic index for characterizing epigenetic landscape.
\begin{figure}[H]
  \centering
  \includegraphics[width=.8\textwidth]{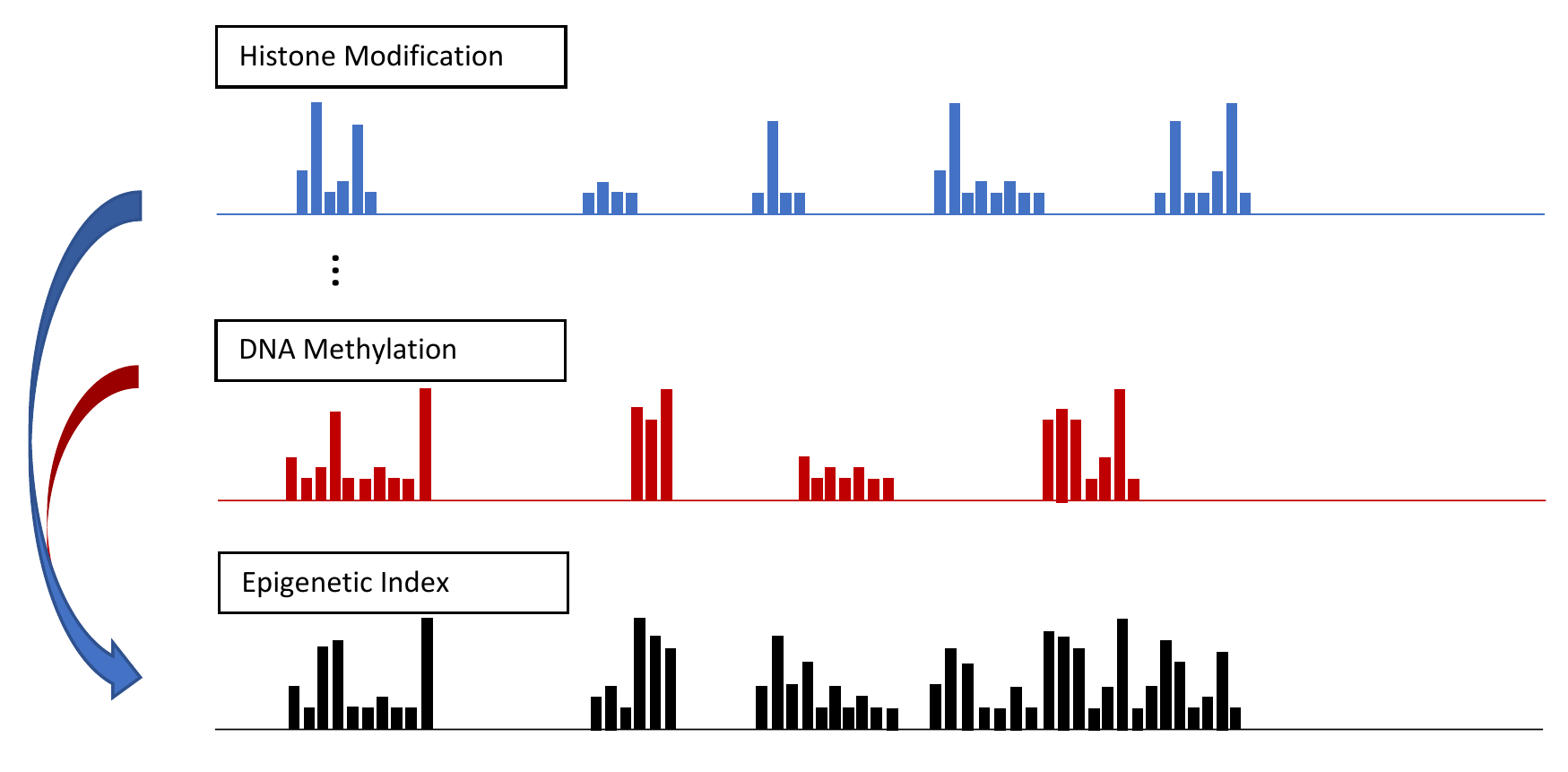}
  \caption{Illustration of data fusion for epigenetic data from the Roadmap Epigenomics Project. As epigenetic indicators, multiple histone modifications and DNA methylation levels (represented by the heights of the bars) across the genome are measured from the experiments. To gain a systematic understanding of the epigenetic landscape, we aim to fuse those measurements and construct an epigenetic index. Such an index is critical for understanding the epigenetic regulation of gene expression systematically.}
  \label{fig:epiindex}
\end{figure}
\end{example*}

Data fusion can have numerous meanings. In this paper, we refer to the process of integrating information from multiple data sources as data fusion \citep{hall1997introduction,hall2004mathematical,waltz1990multisensor,khaleghi2013multisensor}. More specifically, we focus on fusions at the data level. 
The primary motivation of data level fusion is to 
provide key factors that would improve the interpretations and predictions of potential outcomes. Researchers may be able to capture some pieces of the puzzle by a single-data-type study, but data fusion approaches could help fit all pieces together. In the above example, epigenetic modifications, such as histone modification and DNA methylation, jointly characterize epigenetic variations across cell types and genome, especially in the promoter region immediately upstream of where a gene's transcription begins. Integrative analysis of these epigenetic features would considerably improve the power of using epigenetic variation to explain gene expressions. In addition, data fusion approaches could provide more efficient ways of studying the statistical associations than an analysis that uses only a single data type. Finally, since it is computationally intensive to deal with multiple large-scale datasets simultaneously, the integrative analysis that allows data integration can help to ease computation intensity while at the same time preserves as much useful information of the original data as possible.

The integration of multiple data sources has been considered previously in the genomic literature. A comprehensive review of this topic can be found in \cite{ritchie2015methods} and \cite{ray2017adaptive}. The development has taken several paths. One may regard data integration as a dimension reduction problem and try to find a lower dimensional representation of the multiple measurements. Principal component analysis (PCA), for example, is one of the most widely used techniques, which extracts linear relations that best explain the correlated structure in the data. Following this line of thinking, \cite{zang2016high} proposed a PCA-based approach for bias correction and data integration in high-dimensional genomic data. \cite{meng2016dimension} also explored the use of dimension reduction techniques to integrate multi-source genomic data. The second approach resorts to nonnegative matrix factorization to find meaningful and interpretable features in the data. This is the approach taken by \cite{zhang2011novel} to jointly analyze the microRNA (miRNA) and gene expression profiles. \cite{zhang2012discovery} also employed a joint matrix factorization method to integrate multi-dimensional genomic data and identified subsets of messenger RNA (mRNAs), miRNAs, and methylation markers that exhibit similar patterns across different samples or types of measurements. The third approach relies on model-based data integration. For example, \cite{letunic2004smart} developed a database for genomic data integration using the hidden Markov model, which provided extensive annotations for the domains in the genomic sequence. \cite{ernst2012chromhmm} proposed an automated learning system for characterizing chromatin states, which used the multivariate hidden Markov model to learn chromatin-state assignment for each genomic position from multiple sources of epigenetic data.

Due to the special features of multi-source data, there are several challenges in integrating them. First, the multi-source data are often collected using different experimental techniques and thus come in a wide variety of formats or data types. In the above case of characterizing epigenomic landscape, histone modification levels are collected using the ChIP sequencing (ChIP-seq) technique and recorded as continuous measurements, whereas DNA methylation levels are collected using a bisulfite sequencing technique and range from $0$ to $1$. Second, the multi-source data are not necessarily the direct measurements of the same physical quantity, even though they are all manifestations of a certain latent variable. For example, ChIP-seq characterizes the epigenetic activity of the protein tails of histone at different amino acid positions, while whole genome bisulfite sequencing (WGBS) provides single-cytosine methylation levels cross the whole genome. Third, the relationship between the multi-source measurements and the latent variable is often unknown. For these reasons, it is difficult to apply conventional methods such as the PCA to perform the integrative analysis.

To address the above issues, we propose a unified framework for integrating data from multiple sources. Assume there exists a latent variable $Y$, to which the multi-source measurements, say $W_1,\ldots,W_K$, are related in some linear or nonlinear ways $\f_1, \ldots, \f_K$. That is 
\begin{equation}
\f_1(W_1)=Y, \ldots, \f_K(W_K) = Y.
\end{equation}
We seek to find a representation of the latent variable that captures the common information contained in $\f_1(W_1), \ldots, \f_K(W_K)$ even when $\f_1, \ldots, \f_K$ are unknown. In epigenetic data, $W_1,\ldots,W_K$ represent the histone modifications and DNA methylation levels collected from multiple sources (Figure \ref{fig:epiindex}). We aim to find a representation, the epigenetic index, of the latent variable $Y$. Intuitively, when $\f_1, \ldots, \f_K$ are known functions, $Y$ can be represented by averaging $\f_1(W_1), \ldots, \f_K(W_K)$, where each one itself is an estimate of $Y$. In practice, however, we can hardly know $\f_1, \ldots, \f_K$, let alone to estimate $\sum_{k=1}^K \f_k(W_k)/K$. Instead, we propose to find a set of transformation functions $h_1, \ldots, h_K$ such that the distance between each transformation $h_k(W_k)$ ($k=1,\ldots,K$) and their average $\sum_{k=1}^K h_k(W_k)/K$ is minimized. Then the solution to the above optimization problem provides a group of optimal transformations. Their average naturally can be used as a representation for the latent variable. In other words, we intend to find a one-dimensional representation of the higher dimensional data, which is pertinent to a multidimensional scaling problem.

To solve the optimization problem, a technical aspect of our method is to use the basis expansion. It enables the characterization of nonlinear transformations $h_1,\ldots,h_K$ and further simplifies the optimization procedure. With basis expansion, the optimal solution takes a specific form, which makes our method convenient and efficient to implement. To reflect its relationship with multidimensional scaling and basis expansion, we name our method as B-scaling and the one-dimensional representation as B-mean. The asymptotic property of B-mean is also established to provide theoretical underpinnings for our proposed method. 
To illustrate its empirical performance, we applied the B-scaling method to integrate multi-source epigenetic data and established an epigenetic index (EI) across the genome using B-mean. The EI showed better explanations of the changes in gene expressions than any of the data source.

The rest of the paper is organized as follows. In Section \ref{sec:method}, we introduce the proposed method for integrating multiple measurements. In Section \ref{sec:theoretic}, we develop the theoretical properties of the B-mean. 
Simulation studies and real data analysis are reported in Section \ref{sec:sim} and Section \ref{sec:app}. Section \ref{sec:dis} concludes the paper with a discussion. All the proofs are provided in Appendix.

\section{Model Setup}\label{sec:method}

In this section, we first consider the population version of the B-scaling method and then generalize to its sample version. Assume that $Y$ is a one-dimensional random variable, which we do not observe. Instead, we observe $K$ numbers, denoted by $ W = (W_1,\ldots,W_K)^T \in \mathbb R^{K}$, which measure $Y$ in some linear or nonlinear ways. Our goal is to extract $Y$ from $W$.
In the proposed framework, we assume that the latent variable $Y$ and $W_1,\ldots,W_K$ are connected through some functions $\f_1,\ldots, \f_K$. Specifically, let $W_1,\ldots,W_K$ be $K$ random variables and $\mathcal H_k$ be a Hilbert space of functions of $W_k$ for $k=1,\ldots,K$. We assume there exist $\f_k \in \mathcal H_k$ such that $ Y = \f_k(W_k)$ for $k=1,\ldots,K$. 

To recover the transformations $\f_1, \ldots, \f_K$, we can use the fact that, if such transformations do exist, then it must be true that
\begin{equation*}
\f_1(W_1) = \ldots = \f_K(W_K) = Y.
\end{equation*}
Consequently, $\f_1,\ldots,\f_K$ must satisfy the relation
\begin{equation*}
\E \Big[ \sum_{k=1}^K \Big(\f_k(W_k) - K^{-1}\sum_{k=1}^K \f_k(W_k) \Big)^2 \Big] = 0.
\end{equation*}
In other words, if $\sH_1, \ldots, \sH_K$ are families of functions rich enough to contain $\f_1,\ldots, \f_K$, then $(\f_1,\ldots,\f_K)$ should minimize
\begin{equation*}
L(h_1,\ldots,h_K) = \E \Big[ \sum_{k=1}^K \Big(h_k(W_k) - K^{-1}\sum_{k=1}^K h_k(W_k) \Big)^2 \Big]
\end{equation*}
among $(h_1,\ldots, h_K)$.

However, note that the function is minimized trivially by $(h_1, \ldots, h_K)=(0,\ldots,0)$, which yields $L(0,\ldots,0)=0$. To avoid this trivial solution, we assume that $h_1 \neq 0, \ldots, h_K\neq 0$. Furthermore, since the latent variable $Y$ is unaffected by a multiplicative constant, we assume, without loss of generality, $\| h_1 \| = \ldots = \|h_K\|=1$. Based on these observations, we propose the following optimization problem
\begin{eqnarray}\label{eq:obj}
&& \text{minimize} \quad L(h_1,\ldots, h_K), \nonumber \\ 
&& \text{subject to:} \quad h_1 \in \sH_1, \ldots, h_K \in \sH_K, \quad \|h_1\| = \ldots = \|h_K\|=1 .
\end{eqnarray}

\begin{figure}[H]
  \centering
  \includegraphics[width=.9\textwidth]{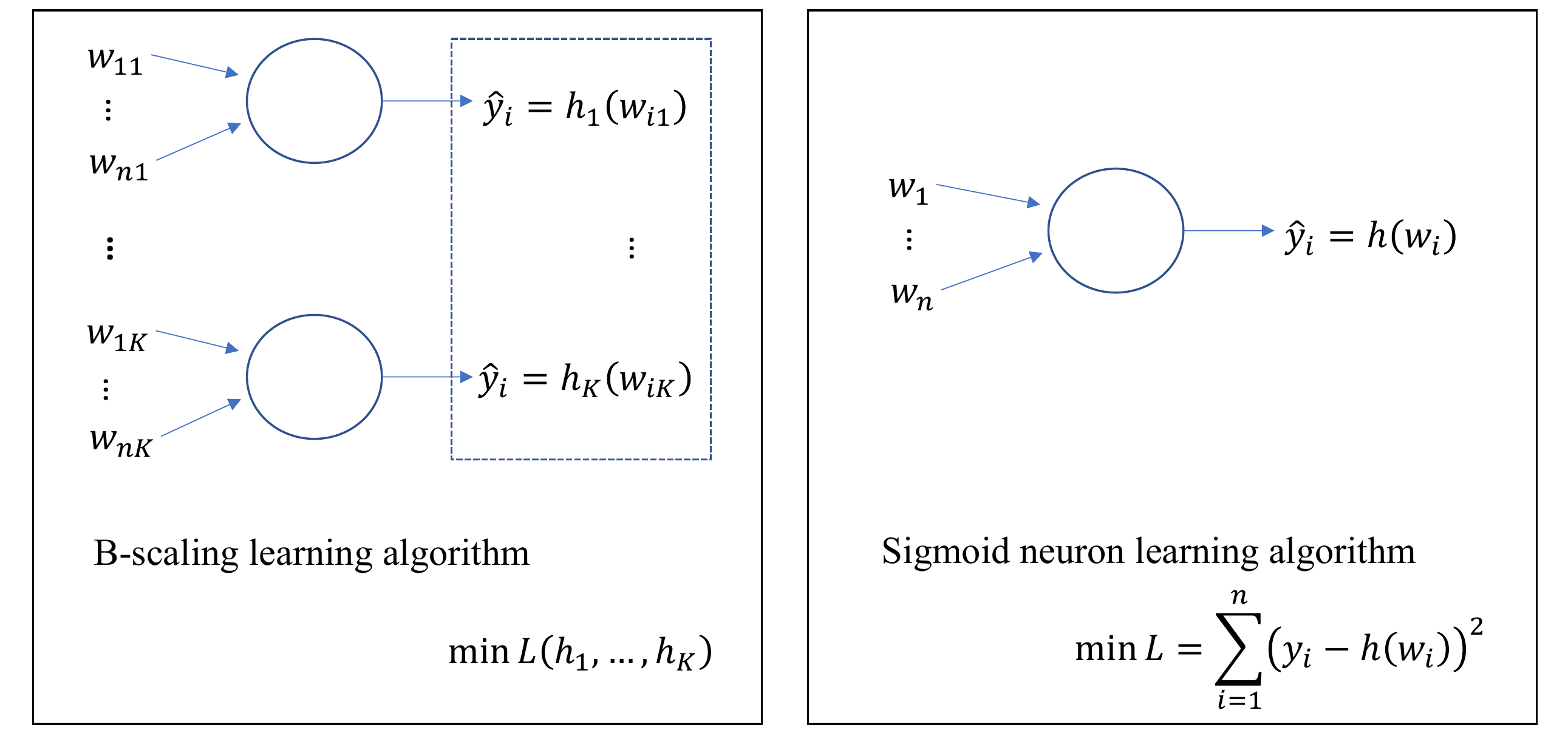}
  \caption{B-scaling from the perspective of sigmoid neuron – the building block of deep neural network. For the left panel, $w_{ik}$ is the $i$th observation for $k$th measurement, and $\hat{y}_i$ is the $i$th fused value for $i=1,\ldots,n$ and $k=1,\ldots,K$.}
  \label{fig:sigmoid}
\end{figure}
\begin{remark}
The ideas of B-scaling and the sigmoid neuron share some common characteristics (Figure \ref{fig:sigmoid}). A sigmoid neuron takes real-valued data as inputs, and the output function, known as the sigmoid function, is usually a nonlinear function. B-scaling also takes multi-source data as inputs. The output functions $h_1, \ldots, h_K$ are defined to be either linear or nonlinear. Furthermore, both methods aim to learn the output functions through minimizing their loss functions. However, B-scaling is performed in an unsupervised manner and allows the output functions to have more general forms. To find those optimal transformations (output functions), B-scaling utilizes the fact that they produce the same output values, which we call as the fused measurement. While the optimization for a sigmoid neuron is implemented in a supervised way.
\end{remark}


\begin{remark}[B-scaling, measurement error model, and multidimensional scaling]
Despite its appearance, this is not a measurement error problem \citep{fuller2009measurement,carroll2006measurement}, because these measurements are not necessarily measuring the same physical quantity. Instead, it is akin to a multidimensional scaling problem \citep{cox2000multidimensional,kruskal1964multidimensional,torgerson1952multidimensional}: that is, we would like to find a one-dimensional representation of higher dimensional data. However, different from multidimensional scaling, we are not
interested in preserving inter-variable distances but rather making a connection among several measurements based on the fact that they are related to the same latent quantity. A technical aspect of our method is the use of basis expansion. Thus, we call our method B-scaling, to reflect its relationship with both multidimensional scaling and basis expansion.
\end{remark}

Let $ (f_1, \ldots, f_K) $ be the solution to the optimization problem (\ref{eq:obj}). In the ideal situation where $\f_k \in \sH_k$ for $k=1,\ldots,K$, the objective function can be perfectly minimized to $0$, and this gives the unique optimizer $(f_1, \ldots, f_K) = (\f_1, \ldots, \f_K)$. In this case, any $f_k(W_k)$ can be used as the representation of the latent variable $Y$. However, in practice we do not make such a strong assumption and allow $\f_k$ to be outside of $\sH_k$. In this case, we use the average of the optimal solution $(f_1, \ldots, f_K)$ as a representation of the latent variable $Y$. This leads to the notions of B-mean and B-variance, as defined below.

\begin{definition}\label{def:1}
 Let $f_1, \ldots, f_K$ be the solution to the minimization problem (\ref{eq:obj}). Let
\begin{equation*}
  \mu_{B}(W_1,\ldots,W_K) = K^{-1}\sum_{k=1}^K f_k(W_k), \ \
  V_B(W_1,\ldots,W_K) = K^{-1} \sum_{k=1}^K [ f_k(W_k) - \mu_{B}(W_1,\ldots,W_K) ]^2.
\end{equation*}
We called them B-scaling mean and B-scaling variance, of $W_1, \ldots, W_K$, respectively, or B-mean and B-variance for short.
\end{definition}

\begin{remark}[B-mean as a representation]
Note that, the B-mean is a random variable, and its definition does not depend on the assumption that the true transformations $\f_k \in \sH_k$ for all $k=1,\ldots,K$. Thus, we have the following two scenarios. 
\noindent\textbf{Scenario 1.} When $\f_1 \in \sH_1, \ldots, \f_K \in \sH_K$, $V_B(W_1,\ldots,W_K)$ is zero, and the unique optimizer $\f_1(W_1), \ldots, \f_K(W_K)$ are equal to the B-mean almost surely. The B-mean thus naturally becomes a representation of $Y$ up to a multiplicative constant.
\noindent\textbf{Scenario 2.} In a more general case, when some $\f_k \notin \sH_k$, the distance between B-mean and the latent variable $Y$ is largely determined by the extent to which $\f_k$ deviates from $\sH_k$. Since we do not make any assumption about the latent variable $Y$ and allow $\f_k$ to be outside $\sH_k$, the B-mean, as defined here, is only a representation of the latent variable $Y$ to the best of our knowledge of observed measurements $W_1, \ldots, W_K$. 
\end{remark}

We call B-mean the fused measurement, and $W_1,\ldots,W_K$ surrogate measurements. We have to emphasize that we do not intend to estimate the latent measurement $Y$. Instead, we are interested in extracting common information contained in muti-source data. Similar to the concept of principal component, B-mean is another type of representation of the data such that its distance with all $K$ transformations is minimized. In the rest of the paper, we focus on discussing the properties of B-mean.

For the choice of $\sH_1, \ldots, \sH_K$, we propose the clans of spline functions as stated by the condition below. It is convenient to express functions in $\sS(m, \tk)$ in terms of spline basis.
\begin{condition}\label{cond:bspline}
  For each $k=1, \ldots, K$, let $\tk$ be the set $\{(\tk_0, \ldots, \tk_{k_0}): 0=\tk_0 < \ldots < \tk_{k_0}=1\}$. Let 
$\sS(m, \tk)$ be the set of spline functions of order $m$ $(m \ge 1)$ with knots $\tk$. That is, for $m=1, \sS(m,\tk)$ is the set of step functions with jumps at the knots; for $m = 2, 3, \ldots$,
  \begin{align*}
  \sS(m, \tk) = &  \{ s \in C^{m-2}[0,1]: s(x) \text{ is a polynomial of degree } \\
  & \hspace{1in}  (m-1) \text{ on each subinterval } [t^{(k)}_i, t^{(k)}_{i+1}] \},
  \end{align*}
where $C^q[0,1]$ is the clan of functions on $[0,1]$ with continuous $q$th derivative. We assume $\sH_k=\sS(m, \tk)$.
\end{condition}
\noindent
For any fixed $m$ and $\tk$, let $\{  N^{(k)}_{1,m}(x),\ldots, N^{(k)}_{\ell,m}(x) \}$ be the basis of $\sS(m, \tk)$, where $\ell=k_0+m-1$ is the number of basis functions in $\sS(m, \tk)$. For any $h_k(x) \in \sS(m,\tk)$, there exits $A_k \in \mathbb R^{\ell}$ such that $ h_k(x) = A_k^T \Nk(x) $, where $\Nk(x) = ( N^{(k)}_{1,m}(x),\ldots, N^{(k)}_{\ell,m}(x) )^T \in \mathbb R^{\ell}$.
By condition (\ref{cond:bspline}), for $k=1, \ldots, K$,
$$h_k(W_k) = A_k^T \Nk(W_k).$$
The objective function in (\ref{eq:obj}) can then be written in matrix as
\begin{equation} \label{eq:objm}
 \E \Big[ \sum_{k=1}^K \Big( A_k^T \Nk(W_k) - K^{-1}\sum_{k=1}^K A_k^T \Nk(W_k) \Big)^2 \Big], A_1, \ldots, A_K \in \R^{\ell}.
\end{equation}
Let 
$A=(A_1^T, \ldots, A_K^T)^T$, $\matN = \text{ diag }(\N_1(W_1), \ldots, \N_K(W_K))$, $\one = (1,1,\ldots,1)^T \in \mathbb R^{K}$, and $Q = I - \mathbf{1} \mathbf{1}^T/K $.
Then the objective function $L(A_1, \ldots, A_K)$ can be rewritten as
\begin{equation*}
L(A) = \E \| \matN^T A - K^{-1} \one \one^T \matN^T A\|^2 = A^T \E(\matN Q \matN^T) A.
\end{equation*}
Thus, at the population level, optimization (\ref{eq:obj}) amounts to solving the following generalized eigenvalue problem:
\begin{align}\label{eq:objpop}
& \text{minimize} \quad A^T \E(\matN Q \matN^T) A \nonumber \\
& \text{subject to:} \quad A^T \var(\matN \one/K) A = 1.
\end{align}
Let $\Lambda = \E(\mathbb N Q \mathbb N^T)$ and $\Sigma = \var( \mathbb N \mathbf{1}/K ) $, the above generalized eigenvalue problem yields the following explicit expression of the B-mean.
\begin{theorem}\label{thm:popb}
 The B-mean of $W_1,\ldots,W_K$ has the form of
  \begin{equation*}
    \mu_B(W_1, \ldots, W_K) = \va^T \mathbb N \mathbf 1/K,
  \end{equation*}
where $\va = \Sigma^{-1/2} \vb$, and $\vb$ is the eigenvector of $\Sigma^{-1/2} \Lambda \Sigma^{-1/2}$ corresponding to its smallest eigenvalue. 
\end{theorem}
\noindent
At the sample level, assume that we observe i.i.d. samples $ \w_{i} = (w_{i1}, \ldots, w_{iK})^T, i=1,\ldots,n$, which measures the i.i.d. latent variable $y_i$ in some linear or nonlinear ways. We use the sample-level counterpart of $\mu_B(W_{1},\ldots,W_{K})$ in Theorem \ref{thm:popb} to estimate the B-mean. Specifically, let 
\begin{align}\label{eq:ldgm}
 & \mathbb N_i = \text{ diag }( \N_1(w_{i1}), \ldots, \N_K(w_{iK}) ) \nonumber \\
  & \Lambda_n = n^{-1} \sum_{i=1}^n \matN_i Q \matN_i^T, \quad \Sigma_n = n^{-1} \sum_{i=1}^n (\matN_i \one - n^{-1} \sum_{i=1}^n \matN_i \one ) (\matN_i \one - n^{-1} \sum_{i=1}^n \matN_i \one )^T/K^2.
 \end{align}
The estimate of $\mu_B( w_{i1}, \ldots, w_{iK})$ is
\begin{equation*}
\hat{\mu}_B(w_{i1}, \ldots, w_{iK}) = \hat \va^T \matN_i \one/K,
\end{equation*}
where $ \hat \va = \Sigma_n^{-1/2} \hat \vb$, and $\hat \vb$ is the eigenvector of $\Sigma_n^{-1/2} \Lambda_n \Sigma^{-1/2}_n$ corresponding to its smallest eigenvalue. In Algorithm \ref{alg:bscaling}, we show an algorithm for computing the B-mean based on the data $(w_{i1}, \ldots, w_{iK}), i=1,\ldots,n$.

\begin{algorithm}[h]
\caption{B-scaling Algorithm}
\label{alg:bscaling}
1. Rescale $\{w_{i1},\ldots,w_{iK}\}_{i=1}^n$ such that they are all defined on $[0,1]$. For $k$th measurement $\{w_{ik}\}_{i=1}^n$ ($k=1,\ldots,K$), generate $\ell$ spline functions of order $m$. Denote the spline functions as $N^{(k)}_{1,m}(x), \ldots, N^{(k)}_{\ell,m}(x)$. 

2. Evaluate the spline functions at each data point of the $K$ measurements, resulting in
\begin{equation*}
 \N_1 (w_{i1}) =
 \begin{pmatrix}
	 N^{(1)}_{1,m}( w_{i1}) \\
	 \vdots \\
	 N^{(1)}_{\ell,m}( w_{i1} )
  \end{pmatrix}
  , \ldots,
  \N_K (w_{iK}) =
  \begin{pmatrix}
  	N^{(K)}_{1,m}( w_{iK}) \\
  	\vdots \\
  	N^{(K)}_{\ell,m}( w_{iK})
  \end{pmatrix}, \quad i=1, \ldots, n.
\end{equation*}
Let $\mathbb N_i = \text{diag} ( \N_1(w_{i1}), \ldots, \N_K(w_{iK}) )$. Calculate $\Lambda_n$ and $\Sigma_n$ by equation (\ref{eq:ldgm}). 

3. Find the eigenvector of $\Sigma_n^{-1/2} \Lambda_n \Sigma_n^{-1/2}$ corresponding to its smallest eigenvalue, denoted as $\hat \vb$. Output $\hat \va = \Sigma_n^{-1/2} \hat \vb$.

4. Estimate the B-mean at the data points $(w_{i1}, \ldots, w_{iK})$ by $\hat \mu_B(w_{i1},\ldots,w_{iK}) = \hat \va^T \mathbb N_i \mathbf 1/K$ for $i=1,\ldots,n$.
\end{algorithm}

\section{Theoretical Properties}\label{sec:theoretic}

In this section, we study the asymptotic property of the B-mean. Recall that the B-mean at $(w_{i1}, \ldots, w_{iK})$ is estimated by $\hat \mu_B(w_{i1}, \ldots, w_{iK}) = \hat \va^T \mathbb N_i \mathbf 1/K$, where $\hat \va = \Sigma_n^{-1/2} \hat \vb$, and $\hat\vb$ is the eigenvector of $\Sigma^{-1/2}_n \Lambda_n \Sigma_n^{-1/2}$ corresponding to its smallest eigenvalue. To show the asymptotic distribution of the B-mean, we first show that both $\hat \vb$ and $\hat\va$ are asymptotically normally distributed.
To simplify presentation, we introduce the following notations. Let $\lambda_{\max}(\cdot)$ and $\lambda_{\min}(\cdot)$ denote the maximum and minimum eigenvalues/singular values of a matrix respectively. 
Let $M_1 \otimes M_2$ denote the Kronecker product of two matrices $M_1$ and $M_2$. Let $M^+$ denote the pseudo-inverse of a matrix $M$. Let $\xrightarrow{D}$ denote in distribution convergence and $\xrightarrow{P}$ in probability convergence.
We make the following assumptions.
\begin{condition}\label{cond:wdis}
Without loss of generality, we assume that $\{w_{i1},\ldots,w_{iK}\}_{i=1}^n$ are i.i.d. random variables in $[0,1]$. Moreover, we assume that each $w_{ik}$ has a finite $4(m-1)$th moment.
\end{condition}

\begin{condition} \label{cond:sigma.eigenvalue} $0 < \lambda_{\min}(\Sigma) \le \lambda_{\max}(\Sigma) < \infty$.
\end{condition}


The finite $4(m-1)$th moment assumption in condition \ref{cond:wdis} guarantees that when the basis $\{\N_1(x),\ldots,\N_K(x)\}$ are of $m$th order, the random sequence $n^{-1} \sum_{i=1}^n N^{k_1}_{\ell_1,m}(w_{ik_1}) N^{k_2}_{\ell_2,m}(w_{ik_2})$ has an asymptotic normal distribution, and it is true for all $k_1,k_2=1\ldots,K$ and $\ell_1, \ell_2=1,\ldots,\ell$. The asymptotic properties of $\Lambda_n$ and $\Sigma_n$ thus follow. Condition \ref{cond:sigma.eigenvalue} ensures that no $N^{(k_1)}_{\ell_1,m}(w_{ik_1})$ has a dominate variance or is linearly dependent on $N^{(k_2)}_{\ell_2,m}(w_{ik_2})$ for $k_1 \neq k_2$ and $\ell_1 \neq \ell_2$. It also implies the existence of $\Sinv$. In the following, we denote $\Sigma^{-1/2} \Lambda \Sigma^{-1/2}$ by $R$, $\Sigma^{-1/2}_n \Lambda_n \Sigma_n^{-1/2}$ by $R_n$, and $\matN \one/K$ by $\matZ$.  With conditions \ref{cond:wdis} and \ref{cond:sigma.eigenvalue}, we have the following asymptotic distribution for $R_n$.
\begin{theorem}\label{thm:R.conv}
  Under conditions \ref{cond:wdis} - \ref{cond:sigma.eigenvalue}, we have
  \begin{align}
    \sqrt{n} [\vectorize( R_n - R )] \xrightarrow{D} N(0, \Pi_{ R}), \label{eq:Rn}
  \end{align}
where
\begin{equation*}
\Pi_R = \begin{pmatrix} \Omega_1 & \Omega_2 \end{pmatrix} \begin{pmatrix} \Phi_{11} & \Phi_{12} \\ \Phi_{21} & \Phi_{22} \end{pmatrix} \begin{pmatrix} \Omega_1^T \\ \Omega_2^T \end{pmatrix},
\end{equation*}
and 
\begin{align*}
& \Omega_1 = -(\Sinvhalf\Lambda \otimes I + I \otimes \Sinvhalf \Lambda)(\St \otimes \Shalf + \Shalf \otimes \St)^{-1} \\
& \Omega_2 = (\Sinvhalf \otimes \Sinvhalf) \\
& \Phi_{11}=\E\Big\{ \Big[\matZ \otimes \matZ - \E(\matZ \otimes \matZ) - (\matZ - \E(\matZ)) \otimes \E(\matZ)  -  \E(\matZ) \otimes (\matZ - \E(\matZ)) \Big] \\
& \quad \quad \quad \quad \Big[\matZ \otimes \matZ - \E(\matZ \otimes \matZ)  -  (\matZ - \E(\matZ)) \otimes \E(\matZ)  -  \E(\matZ) \otimes (\matZ - \E(\matZ)) \Big]^T \Big\} \\
& \Phi_{21} = \E \Big\{ \Big[ (\matN \otimes \matN)\vectorize(Q) - \E(\matN \otimes \matN)\vectorize(Q) \Big] \\
& \quad \quad \quad \quad \Big[ \matZ \otimes \matZ - \E(\matZ \otimes \matZ)  -  (\matZ - \E(\matZ))\otimes\E(\matZ)  -  \E(\matZ) \otimes (\matZ - \E(\matZ)) \Big]^T \Big\} \\
& \Phi_{22} = \E \Big\{  \Big[ (\matN \otimes \matN)\vectorize(Q) - \E(\matN \otimes \matN)\vectorize(Q) \Big] \Big[ (\matN \otimes \matN)\vectorize(Q) - \E(\matN \otimes \matN)\vectorize(Q) \Big] ^T \Big\}.
\end{align*}

\end{theorem}
In the rest of the paper, we write the matrix 
$$\begin{pmatrix} \Phi_{11} & \Phi_{12} \\ \Phi_{21} & \Phi_{22} \end{pmatrix}$$
in Theorem \ref{thm:R.conv} as $\Phi$. To further investigate the properties of the eigenvectors of $R_n$, we need the following condition.
\begin{condition}\label{cond:R.eigenvalue}
 Assume that matrix $ R \in \R^{K\ell \times K\ell}$ has distinct eigenvalues denoted as $d_1,\ldots,d_{r}$, such that $d_{1} > \ldots > d_{r} > 0$, where $r=K\ell$.
\end{condition}
Condition \ref{cond:R.eigenvalue} requires that the eigenvalues of $R$ are separable, which is to guarantee that the eigenvectors are uniquely defined. To conclude, we transform the minimization problem (\ref{eq:obj}) to the generalized eigenvalue problem (\ref{eq:objpop}) with condition \ref{cond:bspline}, and the optimal solution to (\ref{eq:objpop}) is the eigenvector corresponding to the smallest eigenvalue of the matrix $R=\Sigma^{-1/2} \Lambda \Sigma^{-1/2}$. When conditions \ref{cond:sigma.eigenvalue} and \ref{cond:sigma.eigenvalue} are further satisfied, the matrix $R$ is invertible and has separable eigenvalues, which ensures the uniquely defined eigenvectors. Thus, there exists a unique solution to the minimization problem (\ref{eq:obj}) given conditions \ref{cond:bspline}, \ref{cond:sigma.eigenvalue}, and \ref{cond:R.eigenvalue}. The next lemma gives the asymptotic distributions of $\hat\va$ and $\hat\vb$.

\begin{lemma} \label{lm:ba}
Under conditions \ref{cond:wdis}-\ref{cond:R.eigenvalue}, we have
\begin{equation*}
\sqrt{n}(\hat \vb - \vb) \xrightarrow{D} N \Big(\mathbf 0, (M_{b1}, M_{b2})\Phi (M_{b1}, M_{b2})^T \Big),
\end{equation*}
where $\Phi$ is defined in Theorem \ref{thm:R.conv}, and 
\begin{align*} 
& M_{b1} = -[ \vb^T \otimes V( d_{r} I - \Gamma )^{+} V^T] (\Sinvhalf \Lambda \otimes I + I \otimes \Sinvhalf\Lambda)(\St \otimes \Shalf + \Shalf \otimes \St)^{-1} \\
& M_{b2} = [\vb^T \otimes V( d_{r} I - \Gamma )^{+} V^T](\Sinvhalf \otimes \Sinvhalf),
\end{align*} 
with $\Gamma$ being the diagonal matrix $\text{diag }(d_1,\ldots,d_{r})$, and $V$ being a matrix whose columns are eigenvectors of $R$. The asymptotic distribution of $\hat\va$ is given by
\begin{equation*}
\sqrt{n} ( \hat\va - \va ) \xrightarrow{D} N \Big(\mathbf 0, (M_{a1},M_{a2})\Phi (M_{a1},M_{a2})^T \Big),
\end{equation*}
where
\begin{align*}
& M_{a1} = - (\vb^T \otimes I)(\St \otimes \Shalf + \Shalf \otimes \St)^{-1} \\
& \quad \quad \quad - \Sinvhalf [\vb^T \otimes V( d_{r} I - \Gamma )^{+} V^T](\Sinvhalf \Lambda \otimes I + I \otimes \Sinvhalf\Lambda)(\St \otimes \Shalf + \Shalf \otimes \St)^{-1}, \\
& M_{a2} = \Sinvhalf [ \vb^T \otimes V( d_{r} I - \Gamma )^{+} V^T](\Sinvhalf \otimes \Sinvhalf).
\end{align*}
\end{lemma}
As showed in Lemma~\ref{lm:ba}, $\sqrt{n}(\hat\va -\va)$ is asymptotically normally distributed. Notice that for a given new observation denoted as $(\tilde{w}_{1}, \ldots, \tilde{w}_{K})$, its B-mean estimate is represented by an average of $\N_1^T(\tilde{w}_{1}) \hat\va_1, \ldots, \N^T_K(\tilde{w}_{K}) \hat\va_K$, where $\N_1^T(\tilde{w}_{1}), \ldots, \N^T_K(\tilde{w}_{K})$ are the values of basis functions evaluated at $(\tilde{w}_{1}, \ldots, \tilde{w}_{K})$.
With Lemma~\ref{lm:ba}, we show in Theorem~\ref{thm:normality} that the B-mean estimate follows asymptotic normal distribution.
\begin{theorem}\label{thm:normality}
  Suppose that conditions \ref{cond:wdis} through \ref{cond:R.eigenvalue} are satisfied. For a given new observation $\tilde{w}=(\tilde{w}_{1}, \ldots, \tilde{w}_{K})^T$, let $\mu_B( \tilde{w})$ be the B-mean at $\tilde w$, and let $\hat \mu_B( \tilde{w})$ be the B-mean estimate. Then, as $n \rightarrow \infty$,
  \begin{eqnarray}
    \sqrt{n} [\hat \mu_B( \tilde{w}) - \mu_B(\tilde{w})]
    \xrightarrow{D} N(\mathbf 0, \sigma^2_{\mu}),
  \end{eqnarray}
where 
  \begin{align*}
  \sigma^2_{\mu} = \N(\tilde w)^T (M_{a1},M_{a2})\Phi (M_{a1},M_{a2})^T \N(\tilde w)/K^2,
  \end{align*}
with $\N(\tilde w)$ being $( \N^T_1(\tilde{w}_1), \ldots, \N^T_K(\tilde w_K) )^T$.
\end{theorem}

\section{Simulation studies}\label{sec:sim}
We have conducted comprehensive simulation studies to investigate the empirical performance of the proposed B-mean for representing the underlying latent variable. In the first example, we discussed some implementation issues of the B-scaling method. In the second example, we examined the performance of the B-mean by comparing with that of principal components (PCs) and multidimensional scaling (MDS). Their performances were evaluated by their correlations with the latent measurement. For principal component analysis, we reported the maximum correlations that all the PCs can achieve (PC$_{\max}$). For multidimensional scaling, we choose to map all the data points to a one-dimensional space.

\begin{example}
The selection of basis function, its order, and the knots are key issues in implementing the B-scaling algorithm. 
In practice, we choose the natural spline of order $m=4$. The flexibility of a spline then mainly lies in the selection of knots. 
We let knots be the quantiles of $(w_{1k}, \ldots, w_{nk})$ for $k=1,\ldots,K$. 
In this paper, we focus on the B-scaling method itself and choose the same values of $m$ and $k_0$ for all $K$ different measurements. Then for different $k=1,\ldots,K$, the flexibility of fitting splines only lies in the difference of their quantiles. 
For further studies, we can embed into our algorithm the data-driven methods for fitting splines \citep{he2001data,yuan2013adaptive}. 

In this example, we demonstrate how the number of knots impacts the behavior of the B-scaling method. To generate multiple nonlinear functions automatically, we consider the Logit function 
\begin{equation*}
    g(x)=\frac{1}{1+ \text{exp}[20(x-0.5)] }
\end{equation*}
and let $(w_{i1},\ldots,w_{iK})$ measure $y_i$ in nonlinear ways,
$w_{ik} = \sum_{t=1}^H s_k Z_{kt} \delta_t g(y_i+\epsilon_{ik})$ for $i=1,\ldots, n$ and $k=1,\ldots,K$. We generated $y_i$ from $U(0,1)$ as the latent measurement and the errors $\epsilon_{ik}$ from $N(0,0.1)$. In addition, $s_k$ is a scale parameter generated from $U(-10,10)$, $Z_{kt}$ from uniform distribution $U(-\sqrt{3}, \sqrt{3})$, and $\delta_t = (-1)^{t+1} t^{-\nu/2}$. Parameters $\nu$ and $H$ were set to $2$ and $5$ respectively. Let $n=1000, 2000, 3000$, $K=10, 20, 30$, and $k_0$ vary from $11$ to $25$. The performances of B-scaling mean were evaluated by its correlation with $y_i$. As illustrated in Figure~\ref{fig:varknot}, under different scenarios the estimated B-means maintained high correlations with $y_i$ as we vary the number of knots. In practice, we recommend to utilize the B-variance to select a proper number of knots. That is, we propose to choose the number of knots that generates a smaller B-variance. Intuitively, a proper number of knots would result in a better estimation of the $K$ transformations $f_1,\ldots,f_K$ and further a smaller B-variance among $f_1(W_1),\ldots,f_K(W_K)$. We have implemented this method for selecting the number of knots in simulation studies. \red{In Supplementary Material, we further discussed the possibility of combining the plot of B-variance and the associated eigensystem to determine the number of knots (Supplementary Material, Figure S1) for future studies.}
We also reported the computation time of B-scaling in Figure \ref{fig:cputime}. It only takes seconds to calculate the B-mean when the sample size and the number of measurements are relatively small (Figure \ref{fig:cputime}, left panel). As the sample size $n$, number of measurements $K$, and the number of knots increase, the computation time of B-scaling tends to increase (Figure \ref{fig:cputime}).

\begin{figure}[H]
  \centering
  \includegraphics[width=1\textwidth]{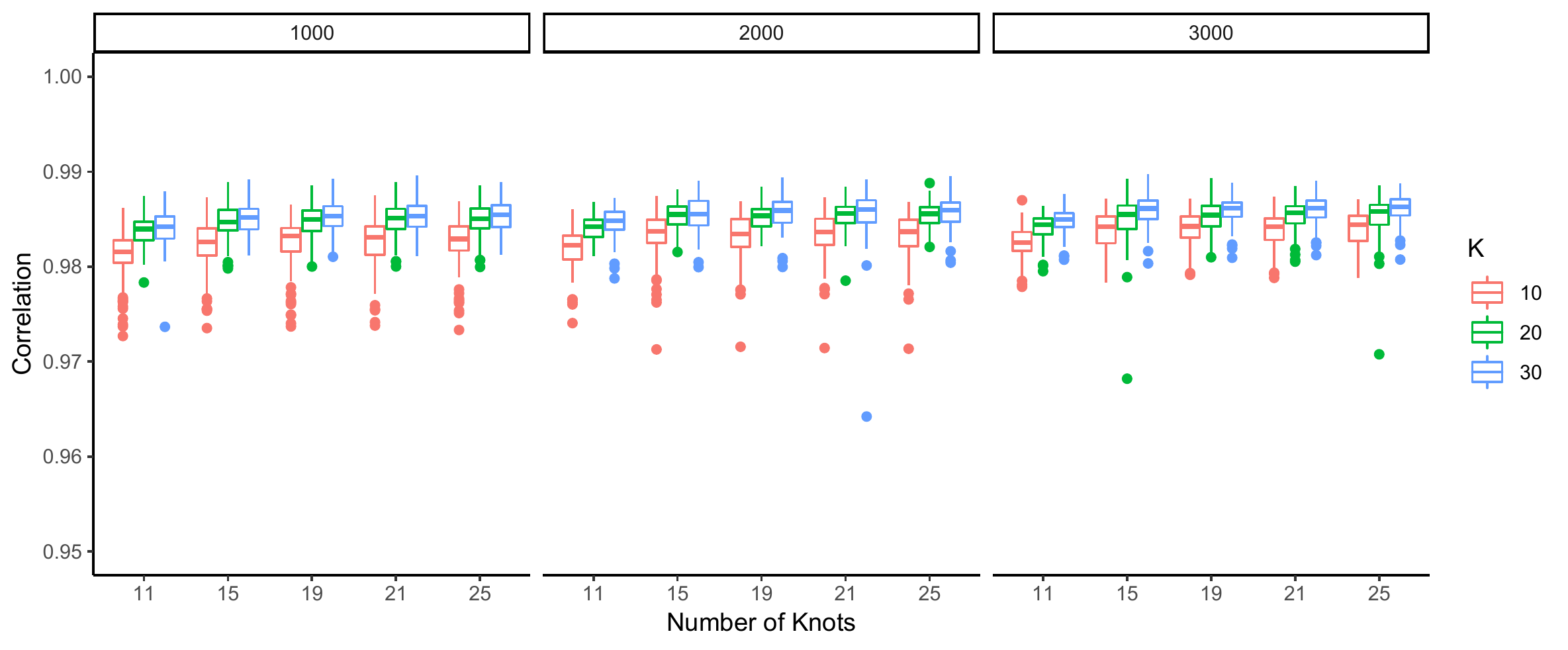}
  \caption{Boxplots of the correlations between the latent variable $y_i$ and estimated B-means.}
  \label{fig:varknot}
\end{figure}

\begin{figure}[H]
  \centering
  \includegraphics[width=1\textwidth]{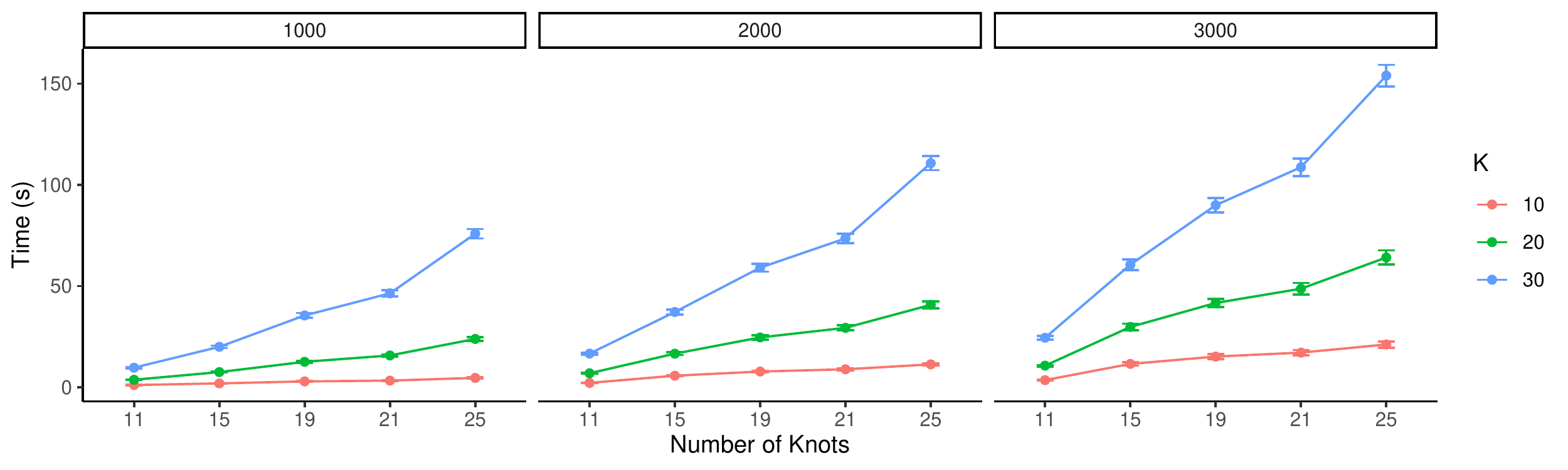}
  \caption{The computation time of B-scaling method with varying $n, K$, and the number of knots. The error bar represents one standard deviation away from the mean.}
  \label{fig:cputime}
\end{figure}

\end{example}

\begin{example} 
In this example, we investigated the performances of the B-mean for representing the latent variable. 
%
We generated $y_i$ as the latent measurement under two scenarios: (1) uniform distribution $U(0,1)$ and (2) normal distribution $N(0,1)$. The $K$ observed measurements $(w_{i1},\ldots, w_{iK})$ as for $i=1, \ldots, n$ were simulated using the same model in Example 1. Let $n=(500,700,1000, 2000, 3000)$ and $K=(7,10, 20, 30)$. We generated $100$ datasets for each combination of $n$ and $K$. To evaluate the performances of the B-mean and PCs, their correlations with $y_i$ in each simulation are reported. The quantities $$\rho_{\max}=\max_k|\text{Corr}(w_{ik}, y_i )|, \quad \bar \rho_0 = K^{-1}\sum_{k=1}^K |\text{Corr}(w_{ik}, y_i)|$$ are also recorded.
Figure \ref{fig:example2} displays the boxplots of the correlations between the latent variable $y_i$ and the aforementioned estimates based on $100$ simulation runs.  The performances of $\rho_{\max}$ and $\bar{\rho}_0$ illustrate that $w_{ik}$s have relatively high correlations with $y_i$ on average. For principal component analysis, we reported the maximum correlations that all the PCs can achieve. 
However, since the PC is a linear combination of the observed measurements $(w_{i1},\ldots, w_{iK})$, its correlation with $y_i$ thus heavily depends on whether $w_{ik}$ is linearly correlated with $y_i$. In all scenarios, the B-mean has the highest correlation on average with the latent variable $y_i$. 

\begin{figure}[H]
  \centering
  \includegraphics[width=1\textwidth]{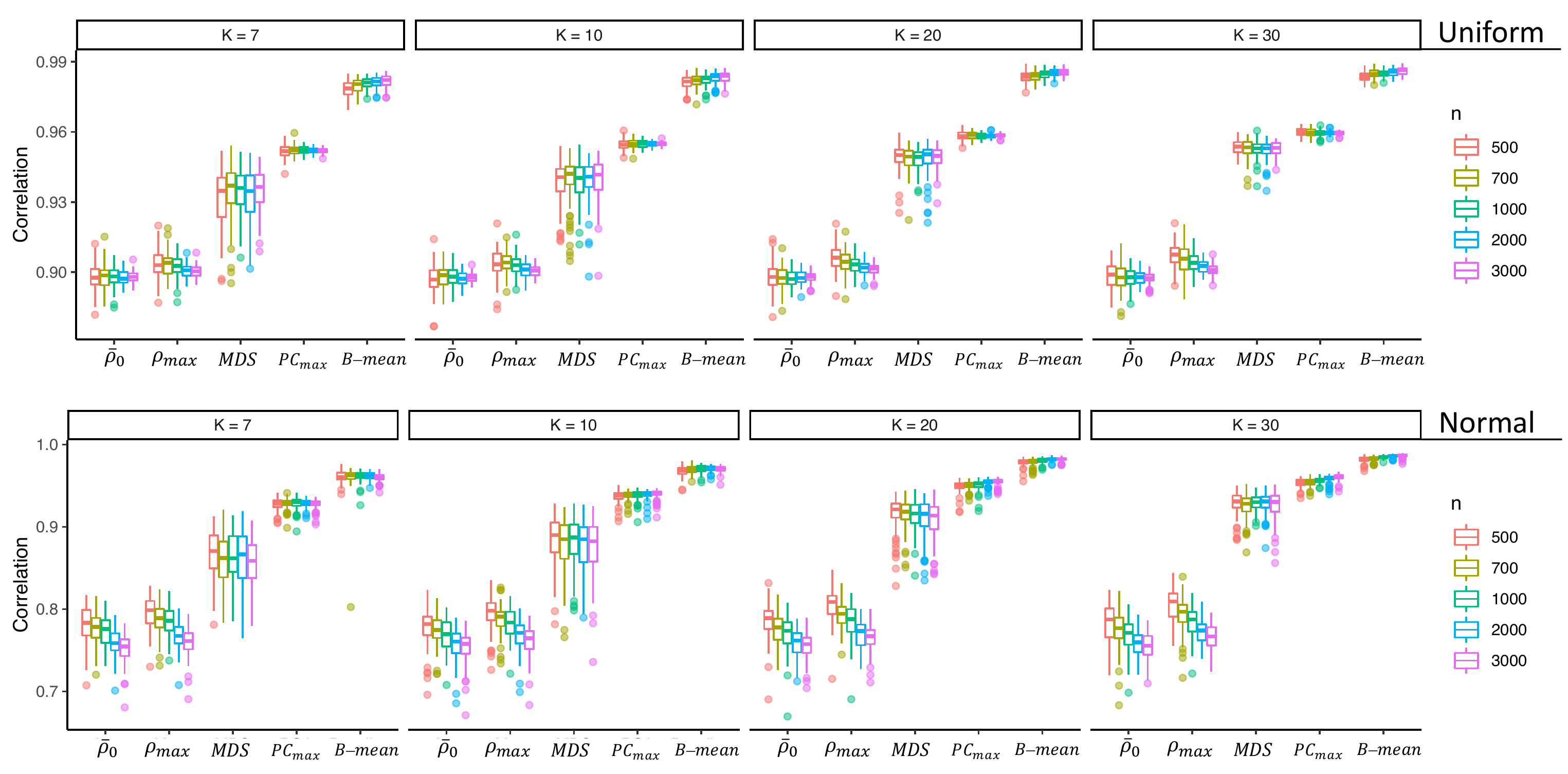}
  \caption{Boxplots of the correlations between the latent variable $y_i$ and the aforementioned estimates in example 2. The latent variable is generated from uniform distribution $U(0,1)$ (upper panel) and standard normal distribution $N(0,1)$ (lower panel).}
  \label{fig:example2}
\end{figure}

\end{example}

\begin{example}
To show the performance of the B-scaling method when $g$ functions have different structures, we consider
\begin{equation*}
    g(x)=\frac{1}{1+ \text{exp}[20(x-0.5)] } \text{ and } g_t(x)=\log(|t/x|),
\end{equation*}
and let $(w_{i1},\ldots,w_{iK})$ measure $y_i$ in nonlinear ways, where $w_{ik} = \sum_{t=1}^H s_k Z_{kt} \delta_t g(y_i+\epsilon_{ik})$ for $k=1,\ldots, \lceil K/2 \rceil$ and $w_{ik} = \sum_{t=1}^H s_k Z_{kt} \delta_t g_t(y_i+\epsilon_{ik})$ for $k=\lceil K/2 \rceil + 1, \ldots, K$. Other settings are the same as Example 2. In all scenarios, B-scaling method outperforms its competitors (Figure \ref{fig:example3}). \red{We further increased the variance of $\epsilon_{ik}$ to 0.3 and generated data pairs with the same settings in Example 2 and 3, respectively. The average correlation $\bar{\rho}_0$ then decreased to a value around $0.65$ in the uniform setting (Supplementary Material Figure S2, A and C) and around $0.4$ in the normal setting (Supplementary Material Figure S2, B and D). Our method still achieved averaged correlations around $0.9$ and $0.7$ under each setting, respectively. } 
\begin{figure}[h]
    \centering
    \includegraphics[width=1\textwidth]{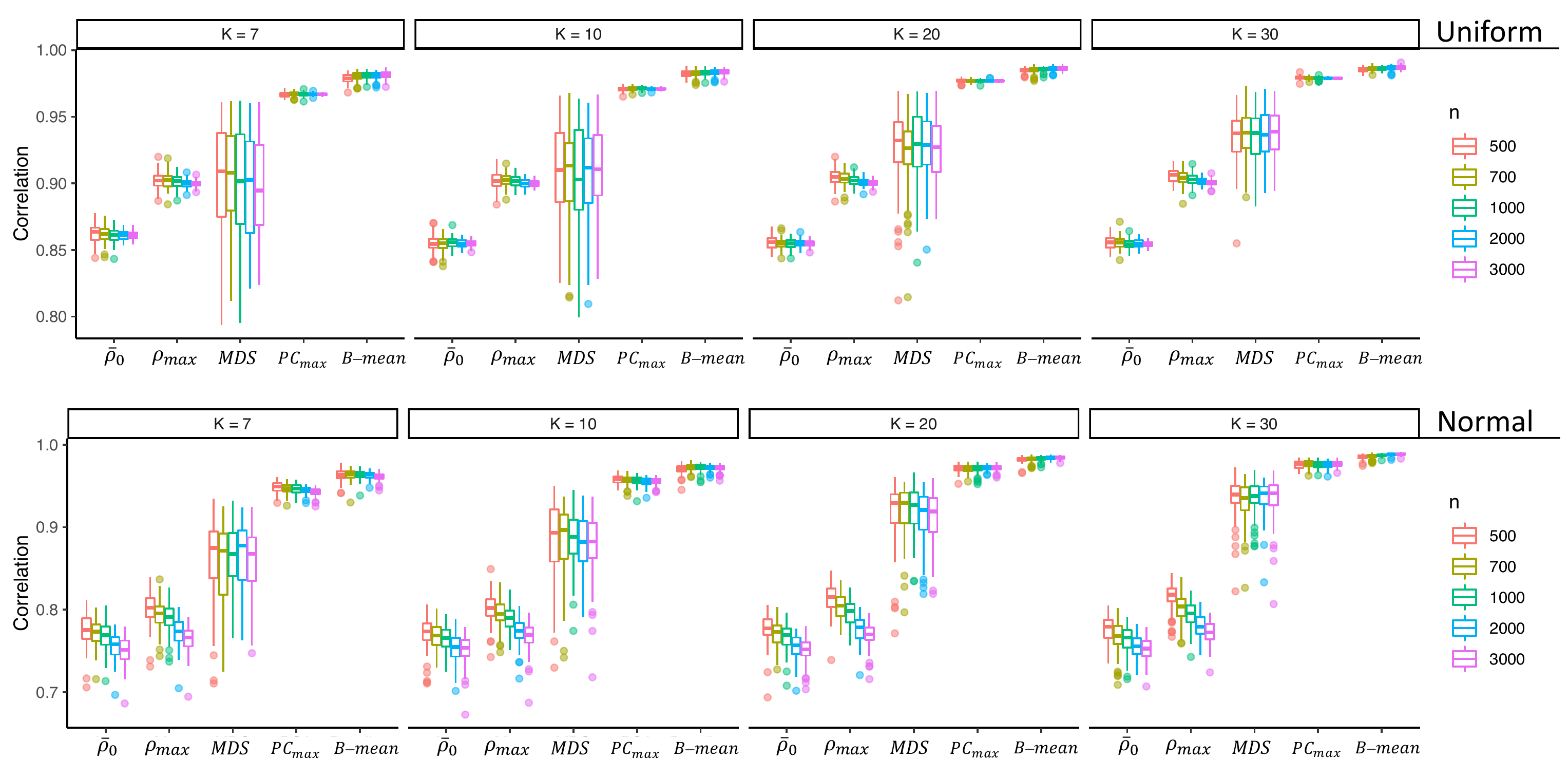}
    \caption{Performance comparisons between B-mean and its competitors in example 3. The latent variable is generated from uniform distribution $U(0,1)$ (upper panel) and standard normal distribution $N(0,1)$ (lower panel).}
    \label{fig:example3}
\end{figure}

\end{example}


\section{Application: Epigenetic Index (EI)} \label{sec:app}

Epigenetic modifications, such as DNA methylation and histone modifications, play an important role in regulating gene expression and numerous cellular processes \citep{portela2010epigenetic}. The modifications target on genetic materials in cells and regulate the biological processes without altering the DNA sequence. Recently, next generation high-throughput sequencing technologies have been adopted to characterize the epigenomic landscapes of primary human tissues and cells \citep{reuter2015high}. The resulting datasets, for instance, the Roadmap Epigenomics database, have provided information on various types of histone modifications and DNA methylation levels across the genome. The integration of these different measures of epigenetic information, also known as the construction of chromatin state maps \citep{ernst2012chromhmm}, for different tissues and cell types under varied conditions can help researchers understand biological processes such as morphogenesis, cell differentiation, and tumorigenesis \citep{sheaffer2014dna, feinberg2016epigenetic}. In this application, the epigenetic activeness is latent and cannot be measured directly. We intend to integrate various types of histone modifications and DNA methylation levels to infer epigenetic activeness across the genome.

We used the human liver cancer cell line data \citep{encode2012integrated}. 
The data consists of the information of six types of histone modifications and whole genome DNA methylation levels measured by ChIP-seq and WGBS respectively. These six histone modifications, i.e., methylation and acetylation to histone proteins H3 at different amino acid positions, include H3K9ac, H3K27ac, H3K4me1, H3K4me2, H3K9me3, and H3K27me3. We used the values of fold change over control for histone modifications and DNA methylation levels on the promoter region, that is, $1000$ base pairs upstream or downstream of where a gene's transcription begins. After filtering the genes with zero histone modification and DNA methylation levels, we selected $4,293$ genes. In summary, the $k$th observed measurement $w_{ik}$ here represents $k$th epigenetic modification level of promoter region for $i$th gene, where $k=1,\cdots,7$, and $i=1,\cdots,4,293$. We aim to integrate these seven epigenetic modifications to infer the latent epigenetic activeness $y_i$. The proposed method, PCA, and MDS were applied to the dataset. To compare the performances of these methods, we considered the following linear regression model and investigated to what extent the fused chromatin EI calculated by these two methods could explain the gene expression levels,
\begin{equation}
\label{eq:real}
\log(g_i) = \alpha_0 + \alpha_1 x_{i} + \epsilon_i,
\end{equation}
where $g_i$ is the expression level (reads per kilobase million, RPKM)  for gene $i$, $i=1,\cdots 4,293$, and $x_{i}$ is the fused EI values corresponding to $i$th gene. 
The adjusted $\text{R}^2$s of the model (\ref{eq:real}) for B-mean, PC$_{\max}$, MDS were $0.321$, $0.214$, and $0.126$ respectively, and the full model with all the predictors, i.e., six histone modifications and DNA methylation, had adjusted $\text{R}^2$ of $0.304$ (Table \ref{tab:real}). In addition, the maximum $\text{R}^2$ was $0.242$ if we only included one of these seven predictors in the model. The B-mean achieves the highest $\text{R}^2$, indicating that the EI fused using B-scaling approach can explain better the changes in gene expression. Since we used the whole genome data and gene expression was influenced by many other factors, the $\text{R}^2$ around $0.3$ was expected \citep{yuan2006statistical}. 
\begin{table}[h]
\centering
\caption{Model performance (adjusted $\text{R}^2$) for different methods}
\label{tab:real}
\begin{tabular}{@{}llllll@{}}
\toprule
Predictor & B-mean & PC$_{\max}$  & MDS & Single   & All \\ \midrule
$\text{Adj.R}^2$ & \bf{0.321}     & 0.214 & 0.126 & 0.242 & 0.304  \\ \bottomrule
\end{tabular}
\end{table}

To further verify the performance of the proposed method, we performed the gene ontology analysis. Through Livercome database \citep{lee2011liverome}, a database for liver cancer genes, we downloaded $3,660$ cancer-related genes. Among these cancer genes, $768$ genes were shown in our selected gene list of size $4,293$. We then divided our gene list into ten subsets based on ten quantiles, which were equally spaced from zero to one hundred, of B-mean or PC values. As PC$_{\max}$ outperforms MDS, we only compared the proposed method with PC$_{\max}$. In Figure \ref{fig:real}, we showed how the number of cancer-related genes from Livercome database in each subset changed as the B-mean (left panel) or PC$_{\max}$ (right panel) values increased. The proposed method showed clear advantage over PC$_{\max}$. When the B-mean values grew, the number of cancer-related genes also increased. This indicated that we could find more cancer-related genes when the genes had higher B-mean values, i.e., higher EI. The Pearson correlation coefficients in Figure \ref{fig:real} also indicated that the B-mean values and the number of liver cancer genes had a stronger linear trend compared to the PCA approach.   

\begin{figure}[H]
  \centering
  \includegraphics[width=.7\textwidth]{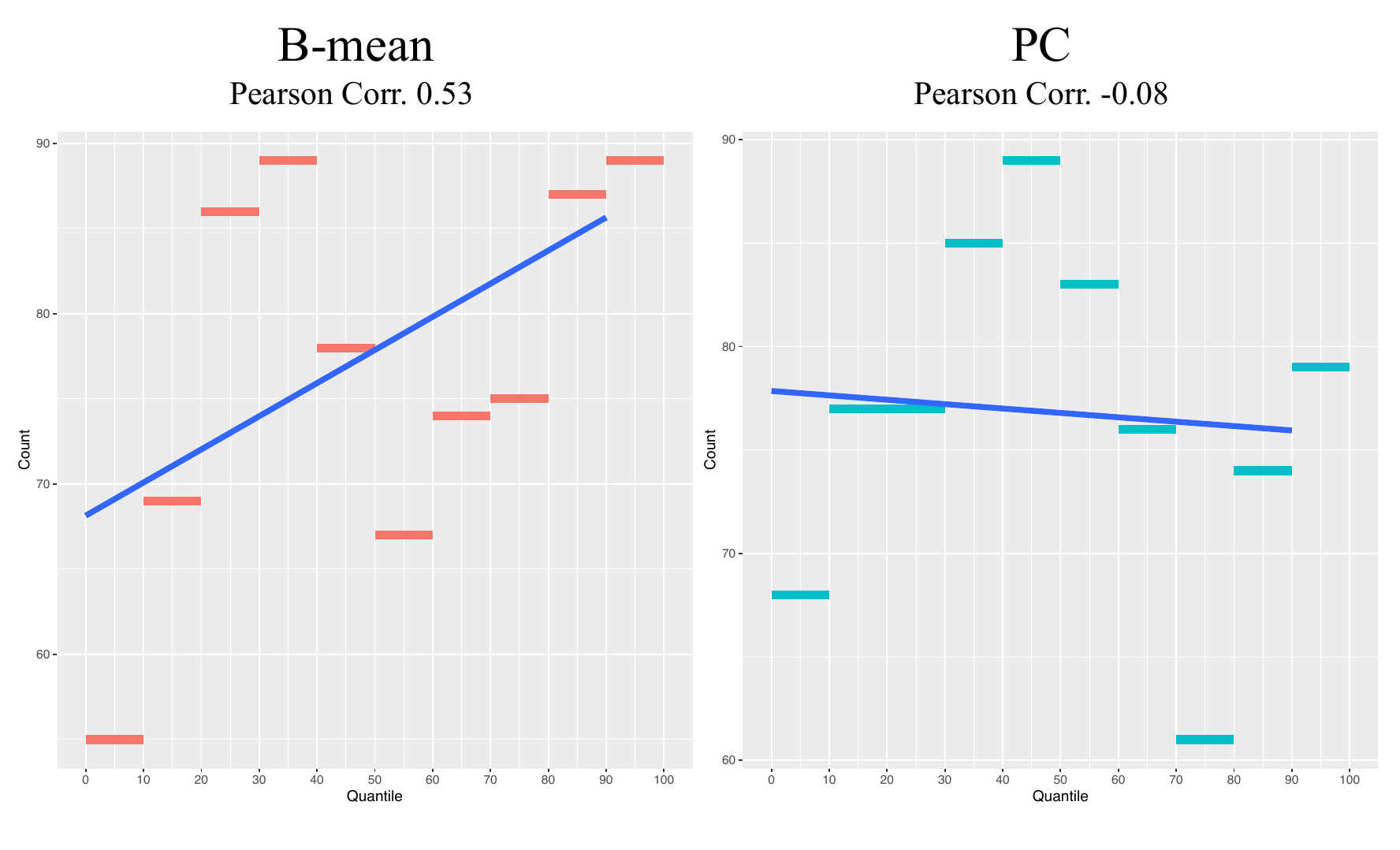}
  \caption{The relationship between the number of liver cancer-related genes and quantiles of B-mean (left panel) or PC$_{\max}$ (right panel) values. The y-axis is the number of liver cancer-related genes, and the x-axis is the corresponding quantiles. The solid blue lines are fitted using the simple linear regression.}
  \label{fig:real}
\end{figure}

\section{Discussion}\label{sec:dis}

The B-scaling approach provides a flexible framework for integrating multi-source data. It adopts a basis expansion step to approximate the nonlinear transformations of $W_1, \ldots, W_K$, and an optimization step to search for a one-dimensional surrogate of those measurements. The basis expansion step enables the characterization of nonlinear transformations, while the optimization step, borrowing strength from the PCA approach, is both convenient and efficient to implement. 
As a trade-off, the B-scaling approach imposes a few assumptions to ensure its efficiency and effectiveness in data fusion. 
In this approach, we only consider a fixed number of basis, order, and number of knots, as stated in condition 2.2. Condition 2.2 is only a working assumption. It simplifies the optimization process and yields a nice and simple optimal solution, although such an assumption is restrictive. It is possible to relax condition 2.2 and allow the number of knots to grow with the sample size. The asymptotic properties of B-mean then may depend on the order between the number of knots and the sample size. 
Another way to make the method fully nonparametric is to impose penalties to the functions $h_1,\ldots,h_K$. The objective function thus becomes
\begin{equation*}
L(h_1,\ldots,h_K) = \E \Big[ \sum_{k=1}^K \Big(h_k(W_k) - K^{-1}\sum_{k=1}^K h_k(W_k) \Big)^2 \Big] + \lambda \sum_{k=1}^K \int h''_k(x)^2 dx,
\end{equation*}
which may further be simplified to have the following quadratic form,
\begin{equation*}
L(h_1,\ldots,h_K) = A^T [\E(\matN Q \matN^T) + \lambda D] A ,
\end{equation*}
where $D$ is a block diagonal matrix with $k$th diagonal element be $D_k$, and $D_k$ is a $\ell$-by-$\ell$ kernel matrix that measures the roughness of $h_k$. 
For any given smoothing parameter $\lambda$, this optimization problem becomes an eigenvalue problem. These ideas deserve to be further explored in future studies.
Regarding the number of knots, it should depend on both the sample size and the true transformations between the latent variable and the observed measurements. With larger sample size and more complex transformations, we recommend to use more knots. While the suggestion for the number of knots in our manuscript is only a naive recommendation, it works well under various scenarios. In this paper, we propose to choose the number of knots with a smaller B-variance. 

The robustness of the B-scaling method is another issue. If one of the $K$ measurements severely deviates from the true measurement, it would have an impact on the minimization process. We could use the following ideas to improve the robustness of the B-scaling method. First, the B-variance can be used to identify a potential poor measurement, because removing it would result in a dramatic decreasing in the B-variance. Thus, if we are able to obtain the information about the accuracy of each measurement, we may put weights on the measurements accordingly to down weight poor measurements. The objective function can then be set as $\E \Big[ \sum_{k=1}^K \delta_k \Big(h_k(W_k) - K^{-1}\sum_{k=1}^K h_k(W_k) \Big)^2 \Big]$, where $\delta_k$ is the weight for measurement $k$. Second, to make this optimization process more robust, we could also change the $L_2$-norm to $L_1$-norm in the objective function, which deserves further investigation.

Regarding the type of transformation functions, we choose to treat the monotone transformations as a special case for the B-scaling method rather than explicitly imposing it as a constraint. If the transformation for $W_k$ is indeed monotone, then the optimal procedure would pick up this information automatically. Since the B-scaling method is based on a minimization problem that is transformed into an eigenvalue problem, imposing an extra constraint would disrupt this simplicity and add significant complexity to the computation with only limited benefit and efficiency. For a more refined version of B-scaling, one might impose some inequality constraints on the spline coefficients to control the trend of the transformations.

Currently, we only considered $W_k (k=1,\ldots,K)$ as a univariate random variable and proposed using the B-scaling mean as a representation of the $K$ different measurements. When it comes to integrating vector-valued samples, we could generalize $W_k$ to be a $p$-dimensional random vector, meaning that $p$ features for each subject are observed in measurement $k$. Correspondingly, we may reformulate the objective function as
\begin{equation*}
L(h_1,\ldots,h_K) = \E \Big[ \sum_{k=1}^K || h_k(W_k) - K^{-1}\sum_{k=1}^K h_k(W_k) ||_2^2 \Big],
\end{equation*}
where $W_k = (W_{k1}, \cdots, W_{kp})^T \in \mathbb R^p$, and each $h_k$ is a multivariate smooth function. The multivariate smooth function can be constructed through tensor product smoothing. For instance, if $p=3$, then $h_k$ can be represented as 
\begin{equation*}
    h_k (W_k) = \sum_{ijl}\beta_{ijl}B_{1i}(W_{k1})B_{2j}(W_{k2})B_{3l}(W_{k3}),
\end{equation*}
where the $\beta_{ijl}$ are coefficients, and $B_{1i}(W_{k1})$, $B_{2j}(W_{k2})$, and $B_{3l}(W_{k3})$ are the basis functions.  
The objective function can be further simplified to have a quadratic form. Such a generalization is an interesting topic for future investigation.

In summary, we believe that the B-scaling method has a broad and important impact on applications in many areas. To facilitate the method development in this direction, we implemented the B-scaling algorithm using programming language R and is available on \href{https://github.com/yiwenstat/b-scaling.git}{GitHub}.

\begin{appendix}
\newcommand{\vep}{\varepsilon}
\section*{Proof}\label{appA}
In this appendix we prove the asymptotic results in Section \ref{sec:theoretic}. We rely on what is known as the von Mises expansion to perform this task. For a comprehensive treatment on this topic, see \cite{fernholz2012mises}. The same approach was used in \cite{li2007directional}. See also \cite{li2018linear}.

Let $\mathcal D$ be the clan of all distributions of $W$. Let $F_n$ be the empirical distribution of $W$, and let $F_0$ be the true distribution of $W$. Let $\mathcal M$ be a metric space. In our context, we take $\mathcal M$ to be the class of all matrices of given dimensions. A matrix-valued statistical functional is a mapping $T: \mathcal D \rightarrow \mathcal M$. When we use the statistic $T(F_n)$ to estimate the parameter $T(F_0)$, we have the following asymptotic result: if $T(F)$ is Fr{\'e}chet Differentiable at $F_0$, then
\begin{equation}\label{eq:expansion}
T(F_n) = T(F_0) + \E_n T^* + o_p(n^{-\half}),
\end{equation}
where $T^*$ is the influence function of $T(F_n)$ satisfying $\E T^* =0 $. Moreover, the variance matrix of $\vectorize(T^*)$ has finite entries. The expansion (\ref{eq:expansion}) is called the first-order von Mises expansion of $T(F_n)$. By the Central Limit Theorem and Slutzky's theorem, we have
\begin{equation*}
\sqrt{n} [ \vectorize(T(F_n)) - \vectorize(T(F_0))] \xrightarrow{D} N(0, \var[\vectorize(T^*)] ).
\end{equation*}

The influence function of a statistical functional is defined as follows. Let $w$ be a point in the support of the random vector $W$. Let $\delta_w$ be the Dirac measure at $w$. Then the influence function of $T(F_n)$ is defined as the derivative
\begin{equation}
T^* = \frac{d}{d\varepsilon} T[ (1-\varepsilon)F_0 + \varepsilon \delta_w]|_{\varepsilon=0}.
\end{equation}

The statistics in Section \ref{sec:theoretic} such as the $R_n$ in Theorem \ref{thm:R.conv} and $\hat\mu_B(\tilde w)$ in Theorem \ref{thm:normality} are all special cases of statistical functionals and their asymptotic distributions can all be derived in this unified fashion.

\subsection{Influence functions of $\Lambda_n, \Sn$ and $\Sninvhalf$}
\begin{proof}
The influence function of $\ldn$ is simple: recall that $\ldn = \E_n(\matN Q\matN^T)$. So we have the following expansion:
\begin{equation*}
\ldn = \ld + \E_n(\matN Q \matN^T - \ld).
\end{equation*}
The influence function of $\ldn$ is simply
\begin{equation}\label{eq:ld_star}
\ld^* = \matN Q \matN^T - \ld.
\end{equation}

We now derive the influence function of $\Sn$. Recall that
\begin{align*}
\Sn &= \var_n(\matN \one /K) \\
&= \E_n[\matN (\one \one^T/K^2) \matN^T] - \E_n(\matN \one /K)\E_n(\one^T \matN^T /K).
\end{align*}
To calculate the influence function of $\Sn$, let $G_{\varepsilon}=(1-\varepsilon)F_0 + \varepsilon\delta_w$, and
\begin{equation*}
\St(\varepsilon) = \int \matN (\one \one^T/K^2) \matN^T d G_{\varepsilon} - \int (\matN \one /K) d G_{\varepsilon} \int (\one^T \matN^T /K) d G_{\varepsilon}.
\end{equation*}
Because $\frac{d}{d\varepsilon} G_{\varepsilon} |_{\varepsilon=0} = \delta_w-F_0$, we have
\begin{align*}
\frac{d \St(\varepsilon)}{d \varepsilon} \Big|_{\vep=0} 
&= \int \matN (\one \one^T/K^2) \matN^T d(\delta_w-F_0) \\
&- [\int (\matN \one /K) d(\delta_w-F_0)]\E(\one^T \matN^T /K) - \E(\matN \one /K)[\int (\one^T \matN^T /K) d(\delta_w-F_0)].
\end{align*}
Thus, the influence function for $\Sn$ is
\begin{align}\label{eq:st_star}
\St^* &= \matN (\one \one^T/K^2) \matN^T  - \E[\matN (\one \one^T/K^2) \matN^T] \nonumber \\
&- [\matN \one /K - \E(\matN \one /K)]\E(\one^T\matN^T /K) - \E(\matN \one /K)[\one^T\matN^T /K - \E(\one^T\matN^T /K)]
\end{align}

Finally, we derive the influence function of $\Sninvhalf$. Since for all $\vep>0$, $\Sinvhalf(\vep) \St(\vep) \Sinvhalf(\vep) = I$, we have
\begin{equation*}
(\Sinvhalf)^* \Shalf + \Sinvhalf \St^* \Sinvhalf + \Shalf(\Sinvhalf)^* = 0.
\end{equation*}
Hence,
\begin{equation*}
(\Shalf \otimes I) \vectorize[(\Sinvhalf)^*] + (I \otimes \Shalf) \vectorize[(\Sinvhalf)^*] = -(\Sinvhalf \otimes \Sinvhalf) \vectorize(\St^*),
\end{equation*}
which implies
\begin{align}\label{eq:vecShalf}
\vectorize[(\Sinvhalf)^*] &= -(\Shalf \otimes I+I \otimes \Shalf)^{-1}(\Sinvhalf \otimes \Sinvhalf) \vectorize(\St^*) \nonumber \\
&= - ( \St \otimes \Shalf + \Shalf \otimes \St)^{-1} \vectorize(\St^*).
\end{align}

\end{proof}
\subsection{Proof of Theorem~\ref{thm:R.conv}.}
\begin{proof}
Since $R_n = \Sninvhalf \ldn \Sninvhalf$, we have
\begin{equation*}
R(\vep) = \Sinvhalf(\vep) \ld(\vep) \Sinvhalf(\vep).
\end{equation*}
Differentiate both sides of this equation with respect to $\vep$, and then evaluate the derivatives at $\vep=0$, to obtain
\begin{equation*}
R^* = (\Sinvhalf)^* \ld \Sinvhalf + \Sinvhalf \ld^* \Sinvhalf + \Sinvhalf \ld (\Sinvhalf)^*.
\end{equation*}
Hence
\begin{align}\label{eq:vecR}
\vectorize(R^*)
&=(\Sinvhalf \ld \otimes I) \vectorize[(\Sinvhalf)^*] + (\Sinvhalf\otimes \Sinvhalf)\vectorize(\ld^*)+(I \otimes \Sinvhalf \ld)\vectorize[(\Sinvhalf)^*] \nonumber \\
&=(\Sinvhalf \ld \otimes I + I \otimes \Sinvhalf \ld)\vectorize[(\Sinvhalf)^*] + (\Sinvhalf\otimes \Sinvhalf)\vectorize(\ld^*) \nonumber \\
&= - (\Sinvhalf \ld \otimes I + I \otimes \Sinvhalf \ld)( \St \otimes \Shalf + \Shalf \otimes \St)^{-1} \vectorize(\St^*) + (\Sinvhalf\otimes \Sinvhalf)\vectorize(\ld^*) \nonumber \\
&= (\Omega_1, \Omega_2)\begin{pmatrix} \vectorize(\St^*) \\ \vectorize(\ld^*)\end{pmatrix},
\end{align}
where $\Omega_1=- (\Sinvhalf \ld \otimes I + I \otimes \Sinvhalf \ld)( \St \otimes \Shalf + \Shalf \otimes \St)^{-1} $ and $\Omega_2=(\Sinvhalf\otimes \Sinvhalf)$. It follows that
\begin{equation*}
\sqrt{n}[ \vectorize(R_n) - \vectorize(R)] \xrightarrow{D} N(0, \Pi_R),
\end{equation*}
where
\begin{equation*}
\Pi_R = (\Omega_1, \Omega_2) 
\begin{pmatrix}
\E[ \vectorize(\St^*) \vectorize^T(\St^*)] & \E[ \vectorize(\St^*) \vectorize^T(\ld^*)] \\
\E[ \vectorize(\ld^*) \vectorize^T(\St^*)] & \E[ \vectorize(\ld^*) \vectorize^T(\ld^*)]
\end{pmatrix}
\begin{pmatrix}
\Omega_1^T \\ \Omega_2^T
\end{pmatrix}.
\end{equation*}

We now compute the moments $\E[ \vectorize(\St^*) \vectorize^T(\St^*)], \E[ \vectorize(\ld^*) \vectorize^T(\St^*)], \E[ \vectorize(\ld^*) \vectorize^T(\ld^*)]$. Let $\matZ = \matN \one /K$. Then, by (\ref{eq:st_star}),
$\St^* = \matZ \matZ^T - \E(\matZ \matZ^T) - [\matZ - \E(\matZ)]\E(\matZ^T) - \E(\matZ)[\matZ - \E(\matZ)]^T.$
It follows that
$\vectorize(\St^*) = \matZ\otimes \matZ - \E(\matZ \otimes \matZ) - [\matZ - \E(\matZ)] \otimes \E(\matZ) - \E(\matZ) \otimes [\matZ - \E(\matZ)]$.
Hence
\begin{align*}
\E[ \vectorize(\St^*) \vectorize^T(\St^*)] = \E\Big\{ &\Big[\matZ \otimes \matZ - \E(\matZ \otimes \matZ) - (\matZ - \E(\matZ)) \otimes \E(\matZ) - \E(\matZ) \otimes (\matZ - \E(\matZ)) \Big] \\
& \Big[\matZ \otimes \matZ - \E(\matZ \otimes \matZ) - (\matZ - \E(\matZ)) \otimes \E(\matZ) - \E(\matZ) \otimes (\matZ - \E(\matZ)) \Big]^T \Big\}.
\end{align*}
Next, by (\ref{eq:ld_star}), $\ld^*=\matN Q \matN^T - \E(\matN Q \matN^T)$. It follows that
$\vectorize(\ld^*) = (\matN \otimes \matN) \vectorize(Q) - \E[ (\matN \otimes \matN) \vectorize(Q)].$
Hence
\begin{align*}
\E[ \vectorize(\ld^*) \vectorize^T(\St^*)] = \E \Big\{ 
& \Big[(\matN \otimes \matN) \vectorize(Q) - \E[ (\matN \otimes \matN) \vectorize(Q)]\Big] \\
& \Big[\matZ\otimes \matZ - \E(\matZ \otimes \matZ) - [\matZ - \E(\matZ)] \otimes \E(\matZ) - \E(\matZ) \otimes [\matZ - \E(\matZ)]\Big]^T\Big\}.
\end{align*}
Finally, by (\ref{eq:ld_star}) again, we have
\begin{align*}
\E[\vectorize(\ld^*) \vectorize^T(\ld^*)] = \E \Big\{ 
&\Big[(\matN \otimes \matN) \vectorize(Q) - \E[ (\matN \otimes \matN) \vectorize(Q)]\Big] \\
&\Big[ (\matN \otimes \matN) \vectorize(Q) - \E[ (\matN \otimes \matN) \vectorize(Q)]\Big]^T \Big\}.
\end{align*}
This completes the proof.
\end{proof}

\subsection{Proof of Lemma~\ref{lm:ba}}

Before proving Lemma \ref{lm:ba}, we introduce the following notations. Let $r=K\ell$ and assume matrices $R$ and $R_n$ have eigenvalue decompositions:
  \begin{eqnarray}\label{eq:eigenR}
    R = V \Gamma V^T, \quad R_n = V_{n} \Gamma_{n}  V_{n}^T,
  \end{eqnarray}
 where $V = (\vv_{1}, \ldots, \vv_{r}) \in \mathbb R^{r \times r}$ and $\Gamma = \text{diag}(d_{1}, \ldots,d_{r}) \in \mathbb R^{r\times r}$, $V_n = (\vv_{n,1}, \ldots, \vv_{n, r}) \in \mathbb R^{r\times r}$ and $\Gamma_{n} = \text{diag}(d_{n,1}, \ldots,d_{n,r}) \in \mathbb R^{r\times r}$. Since $\vb$ and $\hat{\vb}$ are eigenvectors of $R$ and $R_n$ corresponding to their smallest eigenvalues, we have $\vb = \vv_{r}$ and $\hat \vb = \vv_{n,r}$. The following lemma is similar, but slightly stronger than the assertion 3 of Theorem 3.1.8 of \cite{kollo2006advanced}. The proof is also similar to that given in \cite{kollo2006advanced} and is omitted.
 
\begin{lemma}\label{lm:eigenR}
 Suppose $R$ and $R_n$ have eigenvalue decomposition (\ref{eq:eigenR}) and satisfy $\sqrt{n} \vectorize(R_n - R) \xrightarrow{D} N(0, \Pi_R)$ as well as condition \ref{cond:R.eigenvalue}. Then
 \begin{equation*}
 \vv_{n,i} = \vv_i + [\vv_i^T \otimes V (d_i I-\Gamma)^{+} V^T] \E_n \vectorize(R^*) + o_p(n^{-\half}),
 \end{equation*}
where $(\cdot)^{+}$ is the Moore-Penrose inverse.
\end{lemma}

\begin{proof}[{\bf Proof of Lemma \ref{lm:ba}}]
Since $\hat \vb$ and $\vb$ are the $r$th eigenvectors of $R_n$ and $R$ respectively, the influence function of $\hat\vb$ is, by Lemma~\ref{lm:eigenR}, 
\begin{equation*}
\vb^* = [ \vb^T \otimes V (d_r I-\Gamma)^{+}V^T] \vectorize(R^*) \equiv \Omega_3 \vectorize(R^*).
\end{equation*}
Substitute (\ref{eq:vecR}) into the right-hand side, to obtain
\begin{equation}\label{eq:vecb}
\vb^* = (\Omega_3\Omega_1, \Omega_3\Omega_2)\begin{pmatrix}
\vectorize(\St^*) \\ \vectorize(\ld^*)
\end{pmatrix}.
\end{equation}
Hence $\sqrt{n} (\hat\vb - \vb) \xrightarrow{D} N(\mathbf 0, \Pi_{\vb})$, where $\Pi_{\vb} =   (\Omega_3\Omega_1, \Omega_3\Omega_2) \Phi  (\Omega_3\Omega_1, \Omega_3\Omega_2)^T$. The explicit form of $\Omega_3 \Omega_1$ is
\begin{equation*}
-[ \vb^T \otimes V (d_r I-\Gamma)^{+} V^T] (\Sinvhalf \ld \otimes I + I \otimes \Sinvhalf \ld)( \St \otimes \Shalf + \Shalf \otimes \St)^{-1},
\end{equation*}
and the explicit form of $\Omega_3 \Omega_2$ is
\begin{equation*}
[ \vb^T \otimes V (d_r I-\Gamma)^{+} V^T](\Sinvhalf\otimes \Sinvhalf).
\end{equation*}
This proves the first part of the lemma.

Now we turn to the asymptotic distribution of $\hat\va$. Since $\hat \va = \Sninvhalf \hat \vb$, the influence function of $\hat\va$ is 
\begin{align*}
\va^* = (\Sinvhalf)^* \vb + \Sinvhalf \vb^* = (\vb^T \otimes I) \vectorize[(\Sinvhalf)^*] + \Sinvhalf \vb^*.
\end{align*}
Substitute (\ref{eq:vecShalf}) and (\ref{eq:vecb}) into the right-hand side, to obtain
\begin{align*}
\va^* &= - (\vb^T \otimes I)( \St \otimes \Shalf + \Shalf \otimes \St)^{-1} \vectorize(\St^*) + \Sinvhalf\Omega_3\Omega_1 \vectorize(\St^*) + \Sinvhalf\Omega_3\Omega_2 \vectorize(\ld^*) \\
&= [- (\vb^T \otimes I)( \St \otimes \Shalf + \Shalf \otimes \St)^{-1} + \Sinvhalf\Omega_3\Omega_1] \vectorize(\St^*) + \Sinvhalf \Omega_3\Omega_2 \vectorize(\ld^*) \\
&\equiv \Omega_4 \vectorize(\St^*) + \Omega_5 \vectorize(\ld^*)
\end{align*}
It follows that $\sqrt{n} (\hat\va - \va) \xrightarrow{D} N(\mathbf 0, \Pi_a)$, where $\Pi_a = (\Omega_4, \Omega_5) \Phi (\Omega_4, \Omega_5)^T$,
with $\Phi$ being given in Theorem \ref{thm:R.conv}. 
We now derive the explicit forms of $\Omega_4$ and $\Omega_5$. By definition,
\begin{align*}
\Omega_4 &= - (\vb^T \otimes I)( \St \otimes \Shalf + \Shalf \otimes \St)^{-1} + \Sinvhalf\Omega_3\Omega_1 \\
&= - (\vb^T \otimes I)( \St \otimes \Shalf + \Shalf \otimes \St)^{-1} \\
&- \Sinvhalf [ \vb^T \otimes V (d_r I-\Gamma)^{+} V^T] (\Sinvhalf \ld \otimes I + I \otimes \Sinvhalf \ld)( \St \otimes \Shalf + \Shalf \otimes \St)^{-1}.
\end{align*}
Also,
\begin{equation*}
\Omega_5 = \Sinvhalf [ \vb^T \otimes V (d_r I-\Gamma)^{+} V^T](\Sinvhalf\otimes \Sinvhalf).
\end{equation*}
The proof is completed.

\end{proof}

\subsection{Proof of Theorem~\ref{thm:normality}}

\begin{proof}
Recall that $\hat \mu_B(\tilde{w}) = \hat\va^T \matN(\tilde w) \one/K$, where $\matN(\tilde w)$ is the nonrandom matrix 
$$ \text{diag} ( \N_1(\tilde{w}_1), \ldots, \N_K(\tilde{w}_K) ).$$ 
Hence, $ \sqrt{n} [\hat \mu_B(\tilde{w}) - \mu_B(\tilde{w})] = (\hat\va - \va)^T \matN(\tilde w)\one/K$. By the asymptotic distribution of $\hat\va$ in Lemma \ref{lm:ba}, we have
\begin{eqnarray}
  \sqrt{n} [\hat \mu_B(\tilde{w}) - \mu_B(\tilde{w})] \xrightarrow{D} N(0, \sigma_{\mu}^2)
\end{eqnarray}
where $\sigma^2_{\mu} = \one^T \matN^T(\tilde w)(M_{a1}, M_{a2})\Phi (M_{a1}, M_{a2})^T \matN(\tilde w) \one /K^2$, $M_{a1}=\Omega_4$, and $M_{a2}=\Omega_5$.
The vector $\matN(\tilde w)\one$ is simply $\N(\tilde w)$ defined in Theorem \ref{thm:normality}.
\end{proof}

\end{appendix}

\begin{acks}[Acknowledgments]
The authors would like to thank the anonymous referees, an Associate
Editor and the Editor for their constructive comments that improved the
quality of this paper.
\end{acks}

\begin{funding}
Zhong’s research was supported by U.S. National Science Foundation under grants DMS-1440037, DMS-1440038, and DMS-1438957 and by U.S. National Institute of Health under grants R01GM122080 and R01GM113242. Li's research was supported by U.S. National Science Foundation under grant DMS-1713078.
\end{funding}
\begin{supplement}
\stitle{Supplementary Material for ``B-scaling: a novel nonparametric data fusion method"}
\sdescription{We discuss the implementation issues of the B-scaling method and report results from additional simulation studies.}
\end{supplement}

\bibliographystyle{chicago}

\bibliography{bscaling}

\end{document}